\documentclass[11pt,a4paper]{article}
\usepackage[utf8]{inputenc}
\usepackage{cite}
\usepackage{amsmath,amsfonts,amssymb,amsthm}
\usepackage{graphicx}
\usepackage{psfrag}
\usepackage{multicol}
\usepackage{dsfont,mathabx}
\usepackage{fullpage}

\usepackage{bbm}

\usepackage{lmodern}
\usepackage{empheq}
\usepackage{authblk}

\usepackage{dsfont,mathabx}
\usepackage[utf8]{inputenc}
\usepackage{cite}

\newtheorem{definition}{Definition}
\newtheorem{lemma}[definition]{Lemma}

\newtheorem{theorem}[definition]{Theorem}
\newtheorem{corollary}[definition]{Corollary}

\newcommand{\Expec}[2]{\mathbf{E}_{#1} \left[ #2 \right]}
\newcommand{\Prob}[2]{\mathbf{P}_{#1}{\left( #2 \right)}}
\newcommand{\poly}{{\mathrm{poly}}}
\newcommand{\queue}{\mathcal{Q}} 
\newcommand{\confvr}{\mathbf{Q}} 
\newcommand{\confvl}{\mathbf{q}} 
\newcommand{\tetris}{{\sc Tetris}} 
\newcommand{\bigO}{\ensuremath{\mathcal{O}}}

\renewcommand{\leq}{\leqslant}
\renewcommand{\geq}{\geqslant}

\title{\textbf{Self-Stabilizing Repeated  Balls-into-Bins}\thanks{A preliminary 
version of this work appeared in Proc. of the \emph{27th ACM SPAA'15}, 
DOI: http://dx.doi.org/10.1145/2755573.2755584 (see \cite{BCNPP15}).}}

\author[1]{L. Becchetti}
\author[2]{A. Clementi}
\author[1]{E. Natale}
\author[1]{F. Pasquale}
\author[1]{G. Posta}

\affil[1]{\emph{Sapienza} Universit\`a di Roma, {\tt becchett@dis.uniroma1.it}, 
{\tt natale@di.uniroma1.it}, {\tt pasquale@dis.uniroma1.it}, 
{\tt gustavo.posta@mat.uniroma1.it}}
 
\affil[2]{Universit\`a \emph{Tor Vergata} di Roma, {\tt clementi@mat.uniroma2.it}}

\begin{document}

\maketitle

\begin{abstract}
We study the following synchronous process that we call \emph{repeated 
balls-into-bins}. The process is started by assigning $n$ 
balls to $n$ bins in an arbitrary way. In every subsequent round, 
from each non-empty bin one ball is chosen according to some fixed strategy
(random, FIFO, etc), and re-assigned to one of the $n$ bins uniformly at 
random.

\smallskip  We define a configuration \emph{legitimate} if its 
maximum load is $\bigO(\log n)$. We prove that, starting from any 
configuration, the process will converge to a legitimate configuration in linear 
time and then it will only take on 
legitimate configurations over a period of length bounded by \emph{any} 
polynomial in $n$, \emph{with high probability} (w.h.p.). 
This implies that the process is self-stabilizing and that 
every ball traverses all bins in $\bigO(n\log^2 n)$ rounds, w.h.p.

\end{abstract}

\vspace{5cm}

\newpage
\section{Introduction} 
 
We study the following \emph{repeated balls-into-bins} process. 
Given any $n \geqslant 2$,  we initially assign $n$ balls to $n$ 
bins in an arbitrary way.  Then, at every round, from each   
non-empty bin one   ball is  chosen    according to  some   
strategy  (random, FIFO, etc) and re-assigned  to one of the $n$  
bins uniformly at random. Every ball thus performs a sort of 
\emph{delayed} random walk over the  bins and the delays of such  
random walks depend on the size of the  bin queues encountered  during  their   
paths. It thus follows that these random walks 
are correlated. We study the impact of such correlation on the maximum
load.
This process can also be seen as a random-walk based protocol for  
\emph{parallel resource (or task) assignment} in distributed 
systems~\cite{S06,L86}.

Inspired by previous notions of (load) stability~\cite{AKU05,BFG03}, 
we  study the \emph{maximum load} $M^{(t)}$, i.e.,  the maximum 
number of balls inside  one   bin at   round $t$ and we are 
interested in the largest $M^{(t)}$ achieved by the process over a 
period of \emph{any polynomial} length. We say that a configuration 
is  \emph{legitimate}  if its maximum load is $\bigO(\log n)$ and a 
process is \emph{stable} if, starting from any legitimate 
configuration, it only takes on  legitimate configurations over a 
period of $\poly(n)$ length, w.h.p. We also investigate a 
probabilistic version of self-stabilization \cite{D74,D00}: we say that 
a process is \emph{self-stabilizing} if it is stable and if, moreover, 
starting from \emph{any} configuration, it converges to a legitimate 
configuration, w.h.p. The  \emph{convergence time} of a 
self-stabilizing process is the maximum number of rounds required to 
reach a legitimate configuration starting from any configuration. 
This natural notion of (probabilistic) self-stabilization has also 
been inspired  by that  in~\cite{IJ90} for other distributed 
processes.
  
Stability has consequences for   other  important aspects of 
this process. For instance, if the process is stable, we can get 
good upper bounds on  the   \emph{progress} of a ball, namely the 
number of   rounds the ball  is selected  from its 
current  bin queue,       along  a sequence  of $t\geqslant 1$ rounds. 
Furthermore, we can eventually bound the \emph{parallel} cover 
time, i.e.,   the time required for  every   ball to visit \emph
{all} bins. Self-stabilization has  also   important consequences 
when the system is prone to  transient faults \cite{D74,L85,D00}.

To the best of our knowledge, the repeated balls-into-bins process 
was first studied in \cite{BCE10}, where it is used as a crucial sub-procedure 
to optimize the message complexity of a gossip algorithm in the 
complete graph, and then in~\cite{BCN14,EK15}. The analysis 
in~\cite{BCE10,EK15} (only) hold for very-short (i.e. logarithmic) periods, 
while the analysis in~\cite{BCN14} considers periods of arbitrary length but 
it (only) allows to achieve a bound on the maximum load that rapidly 
increases with time: after  $t$ rounds, the maximum load is  
bounded by $\bigO\!\left(\sqrt t \right)$ w.h.p. By adopting the FIFO strategy 
at every bin queue, the latter result easily 
implies that   the  progress of any ball is $\Omega(\sqrt t)$ w.h.p. 
On the other hand, an upper bound $\bigO\!\left( n^2 \log n \right)$ for the 
parallel cover time of the repeated balls-into-bins process easily follows
from the fact that the cover time of one single random walk on the complete
graph is $\Theta(n \log n)$ w.h.p.

Previous results are thus not helpful to establish  whether this 
process is stable (or, even more, self-stabilizing) or not. 
Moreover, the  previous analyses of the maximum load   in \cite
{BCN14,BCE10,EK15} are far from tight, since they   rely on  some 
rough approximations of the studied process via  other, much simpler 
Markov chains: for instance, in \cite{BCN14}, the authors consider 
the   process - which clearly dominates the original one  - where, 
at every round,  a new ball  is inserted in   every  empty bin. 
That analysis thus does not exploit the global invariant 
(a fixed number $n$ of balls) of the original process.

\medskip \noindent \textbf{Our Results.} We provide a new, tight 
analysis of the repeated balls-into-bins process that significantly 
departs from previous ones  and  show that 
the system is self-stabilizing.   
We prove that, for any arbitrarily-large constant $c$, if the process 
starts from a legitimate configuration, then the maximum load $M^{(t)}$
is $\bigO(\log n)$ for all $t = \bigO(n^c)$, w.h.p. 
Moreover, starting from any configuration, the system 
reaches a legitimate configuration within $\bigO(n)$ rounds, w.h.p. 

Our result strongly improves over the best previous 
bounds~\cite{BCN14,BCE10,EK15} and it is almost tight, since the 
classical lower bound $\Omega(\log n /\log\log n)$ on the maximum load 
(see, e.g.,~\cite{MU05}) clearly applies also in our repeated setting. 
Our result further implies that, under the FIFO queueing policy, any ball  
performs $\Omega(t/\log n)$ steps of its individual random walk over 
any sequence of $t = \poly(n)$ rounds w.h.p., so the parallel cover time 
is $\bigO\!\left(n \log^2n\right)$ w.h.p. 
This is only a $\log n$ factor away from the lower bound following from
the single-ball process.

\smallskip
Besides being interesting in their own right, balls-into-bins processes  are 
used to model and analyze   several important randomized protocols 
in parallel and distributed computing \cite{ABKU99,BCSV06,V03}. In 
particular,  the process we study  models a  natural 
randomized solution to the problem of \emph{(parallel) resource 
(or task) assignment}   in distributed systems  (this problem is 
also known as  \emph{traversal}) \cite{S06,L86}. In the basic  
case,  the goal is  to assign   one resource in mutual exclusion to 
\emph{all} processors (i.e. nodes) of a distributed system.  This  
is typically described as a \emph{traversal}  process performed by 
a \emph{token} (representing the resource or task) over the 
network.  The process terminates when the  token   has visited    all 
nodes of the system. Randomized protocols for this problem \cite
{C11} are   efficient  approaches when, for instance, the network 
is prone to faults/changes and/or when there is no  
global labeling    of the nodes.
   
A simple randomized protocol is the one based on \emph{random 
walks} \cite{C11,IJ90,IKOY02}: starting from any  node, the 
token performs a random walk over the network until all nodes are 
visited, w.h.p. The first round in which all nodes have 
been  visited by the token is called the \emph{cover time } of 
the random walk~\cite{C11,LPW09}. The expected cover time for 
general graphs is $\bigO(|V| \cdot |E|)$ (see, for example,~\cite{MU05}).
   
In distributed systems, we often are in the presence of 
\emph{several} resources or tasks that must be processed by every 
node \emph{in parallel}. This naturally leads to consider  the   
parallel version of the basic    problem in which  $n$ different 
tokens (resources) are initially  distributed over the set of nodes 
and every token must visit all  nodes of the network. 
Similarly to the basic  case, an efficient randomized solution is 
the one based on  (parallel) random walks. In order to visit the 
nodes, every token performs  a random walk under the   
constraint that every node    can process and release   at most one 
token  per round. Again, maximum load  is a critical complexity 
measure:  for instance, it can determine the required buffer size at every node, bounds on the token 
progress and, thus, on  the parallel cover time.

It is   easy to see that, when the graph is complete, the 
above       protocol -  based on    parallel random walks -     is 
in fact equivalent to the repeated balls-into-bins process  analyzed 
in this paper. For this case,      our results imply that,  every 
token   visits all nodes of the system with at most a logarithmic 
delay w.r.t. the case of a single token: so,  we can derive an upper 
bound $\bigO(n \log^2 n)$ for the parallel cover time, starting  
from \emph{any} initial configuration. 

We can also consider the adversarial model in which, in some 
\emph{faulty} rounds, an adversary can re-assign the tokens to the nodes 
in an arbitrary way. The self-stabilization and the  linear 
convergence time shown in Theorem~\ref{thm::main} imply that the  
$\bigO\!\left(n \log^2n\right)$ bound on the cover time still holds, provided 
that faulty rounds occur with a frequency no higher than $c n$, for 
a sufficiently large constant $c$.

\medskip 
\noindent 
\textbf{Related Work.}
 
\smallskip 
\noindent -  \emph{Random Walks on Graphs.}  
The repeated balls-into-bins process was  
first considered in~\cite{BCE10,BCN14,EK15}, since it describes the 
process  of performing parallel random walks in the (uniform) 
gossip model (also known as random phone-call model \cite
{DGHILSSST87,KSSV00}) when every message can contain at most  
one token. 
Maximum load (i.e., node 
congestion), token delays, mixing and cover times are here    
the most  crucial aspects. We remark that the flavor of these 
studies is different from ours: indeed, their main goal is 
to    keep   maximum load and token delays logarithmic 
over some \emph{polylogarithmic
period}. Their aim is to achieve a fast mixing time for every 
random walk in  the case of   good expander graphs. In particular, 
in \cite{BCE10},  a logarithmic bound is shown for the complete 
graph when $m=\bigO(n/\log n)$ random walks are 
performed over a logarithmic time interval.   A similar bound is also 
given for some families of almost-regular random graphs in~\cite{EK15}. 
Finally, a new analysis is given in~\cite{BCN14}
for regular graphs yielding the  bound $\bigO\!\left(\sqrt t\right)$.

\smallskip 
\noindent -  \emph{Parallel Computing.} Balls-into-bins 
processes have been extensively studied in the area of parallel and 
distributed computing, mainly  to address  balanced-allocation 
problems \cite{ABKU99,BCSV06,RS98}, PRAM simulation \cite{KLM96} and 
hashing \cite{DGMMPR10}. In order to optimize the total number 
of random bin choices used for the allocation, further allocation 
strategies have been proposed and analyzed (see, e.g., 
\cite{adler1995parallel,BKSS13,mitzenmacher2001power,MPS02,V03}). 
As previously  mentioned, our notion of 
stability is inspired by those studied in~\cite{AKU05,BFG03,Berebrink16}
where load balancing algorithms are analyzed in scenarios 
in which new tasks arrive during the run of the system, and existing 
jobs are executed by the processors and leave the system. 
An adversarial model for a sequential balls-into-bins process has been 
studied in~\cite{AS06}. We remark that, in the above previous works, the goal 
is different from ours: each ball/task  must be allocated to \emph{one, 
arbitrary}  bin/processor (it is not a token-traversal process).

\smallskip 
\noindent - \emph{Queuing Theory}. To the best of our 
knowledge, the closest model to our setting in classical 
queuing theory is the \emph{closed Jackson network} \cite{A03}. 
In this model, time is continuous and each node processes a single 
token among those in its queue; processing each token 
takes an exponentially distributed interval of time. As 
soon as its processing is completed, each token leaves the current node 
and enters the queue of a neighbor chosen uniformly at random. 
Notice that, since time is continuous, the 
process' events are sequential, so that the associated Markov chain 
is much simpler than the one describing our parallel process. In 
particular, the stationary distribution of a closed Jackson 
network can be expressed as a product-form distribution. It is 
noted in \cite{HW92} that ``[\dots] virtually all of the models 
that have been successfully analyzed in classical queuing 
network theory are models having a so-called product form 
stationary distribution''.  Because of the above considerations 
regarding the difficulty of our process (especially the 
non-reversibility of its Markov chain), the stationary 
distribution   is instead very likely not to exhibit a product-form 
distribution, thus laying outside the domain where the 
techniques of classical queuing theory seem effective. 
We finally cite the seminal work \cite{BKRSW01} on 
\emph{adversarial  queing systems}: here, new tokens (having   
specified source and destination nodes) are inserted in the 
nodes according to some adversarial strategy and a  notion of 
\emph{edge-congestion} stability  is investigated.

\section{Self-Stabilization of repeated balls into bins} \label{sec:sstab}
In order to study the 
maximum load of the repeated balls into bins process, the state of the 
system is completely characterized by the load of every bin. 
Formally, for each bin $u \in [n]$ let $\queue_u^{(t)}$
be the r.v.\footnote{We always use capital letters for random 
variables, lower case for quantities, and bold for vectors.} 
indicating the number of balls, i.e. the \emph{load},  in $u$ at 
round $t$. We write $\confvr^{(t)}$ for the vector of these random 
variables, i.e., $\confvr^{(t)} = \left( \queue_u^{(t)}\;:\; u \in 
[n]\right)$. We write $\confvl = (q_1, \dots, q_n)$ for a \emph
{(load) configuration}, i.e., $q_u \in \{0, 1, \dots, n\}$ for every 
$u \in [n]$ and $\sum_{u = 1}^n q_u = n$.  We define the \emph
{maximum load} of a configuration $\confvl = (q_1, \dots, q_n)$ as  
 
\[ 
M(\confvl) \ = \ \max\{ \, q_u \, : \,  u \in [n] \, \} \, 
, 
\] 
and, for brevity' sake, given any round $t$ of the process, we 
define 
 
\[ 
M^{(t)} \ = \ M(\confvr^{(t)})   
\] 
According to the above 
definition, we say that a configuration $\confvl$ is \emph
{legitimate} if $M(\confvl) \leq \beta \cdot \log n$, for some 
absolute constant $\beta > 0$.
 
In this section we prove the  main theorem of this paper.

\begin{theorem} \label{thm::main} Let $c$ be an 
arbitrarily-large constant and let $\confvl$ be any legitimate 
configuration. Let the \emph{repeated balls-into-bins process} start 
from $\confvr^{(0)} = \confvl$. Then, over any period of   
length $\bigO(n^c)$, the process   visits only legitimate 
configurations, w.h.p., i.e., $M^{(t)} = \bigO(\log n)$ for all  $t 
= \bigO(n^c)$ w.h.p. Moreover, starting from any configuration,  the 
system reaches a legitimate configuration within $\bigO(n)$ rounds, 
w.h.p. 
\end{theorem}

\subsubsection*{Overview of the analysis}
In the repeated balls-into-bins process, every bin can release at 
most one ball per round. As a consequence, the random walks 
performed by the balls delay each other and are thus correlated in a 
way that can make bin queues larger than in the independent 
case. Indeed, intuitively speaking, a large load observed at a bin 
in some round makes ``any'' ball more likely to spend several future 
rounds in that bin, because if the ball ends up in that bin in one 
of the next few rounds, it will undergo a large delay. This is 
essentially the major technical issue to cope with.
 
The previous  approach in \cite{BCN14} relies on the fact that, in 
every round, the expected balance between the number of incoming 
and outgoing balls is always non-positive for every non-empty bin 
(notice that the expected number of incoming balls is always at 
most one). This may suggest viewing the  process as  a sort of parallel 
\emph{birth-death} process \cite{LPW09}. Using this approach and 
with some further  arguments, one can (only) get the    
``standard-deviation'' bound $\bigO(\sqrt t)$  in  \cite{BCN14}. Our 
new analysis proving Theorem \ref{thm::main} proceeds along three main 
steps.

\noindent
\emph{i)} We first show  that, after the first round,  the  aforementioned
expected balance is always   negative, namely,  not   larger than 
$-1/4$. Indeed, the number of empty bins remains at least $n/4$
with (very) high probability, which is extremely useful since a bin can 
only receive tokens from non-empty bins. This fact is shown to hold starting 
from \emph{any} configuration and over any period of polynomial length.

\smallskip
\noindent 
\emph{ii)} In order to exploit the above negative  balance to bound 
the load of the bins, we need some  strong concentration bound on  
the number of balls entering a specific bin $u$ along any 
period of polynomial size. However, it is easy to see that, for any 
fixed $u$, the random variables $\left\{Z^{(t)}_u\right\}_{t \geq 0}$ counting 
the number of balls entering bin $u$ are not mutually independent,
neither are they negatively associated, so that we cannot   apply standard 
tools to prove concentration (see Appendix \ref{sec::apx-association} 
for a counterexample).
\noindent
To address  this    issue, we define a simpler repeated balls-into-bins process 
as follows. 

\smallskip
\noindent
\begin{center}
\fbox{
\begin{minipage}{15cm}
{\sc \tetris\ process.} Starting from any configuration with at least $n/4$ empty bins, in each round\\
- from every non-empty bin we pick one ball and we throw it away, and \\ 
- we pick exactly $(3/4) n$ \emph{new balls} and we put each of them independently and u.a.r. in one 
of the $n$ bins.
\end{minipage}
}
\end{center}

\smallskip
\noindent
Using a coupling argument and our previous upper bound on the  number
of empty bins, we prove that the maximum number of balls accumulating in a 
bin in the original process is not larger than the maximum number of balls 
accumulating in a bin in the \tetris\ process, w.h.p. 

\smallskip
\noindent
\emph{iii)}
The \tetris\ process is simpler than the original one since, at 
every round, the number of balls assigned to the bins does not 
depend on the system's state in the previous round. Hence,  
random variables $\left\{\hat Z^{(t)}_u\right\}_{t \geq 0}$ counting the number 
of balls arriving at bin $u$ in the \tetris\ process are mutually independent.
We can thus apply standard concentration bounds. On the other hand, differently 
from the approximating process considered in~\cite{BCN14}, the negative balance 
of incoming and outgoing balls proved in Step i) still holds, thus yielding a 
much smaller bound on the maximum load than that in~\cite{BCN14}.
A probabilistic version of the \tetris\ process, where the number of new balls 
arriving at each round is a random variable with expectation $\lambda n$, for 
some $\lambda = \lambda(n) \in [0,1]$, has been recently studied 
in~\cite{Berebrink16}.

In the remainder of this section, we formally describe the above three 
steps, thus  proving Theorem \ref{thm::main}.

\subsection{On the number of empty bins}
We next show that the number of \emph{empty} bins is at least a constant 
fraction of $n$ over a very large time-window, w.h.p.
This fact could  be    proved by standard concentration arguments if, at every
round, \emph{all} balls were thrown independently and uniformly at random. 
A little care  is instead required in our process to properly handle, at any round, 
``congested'' bins whose load exceeds $1$. These bins will be surely non-empty at the next round 
too. So, the number of empty bins
at a given round also depends on the number of congested bins in  the previous round. 

\begin{lemma}\label{lemma:emptyqueuesround}
Let $\confvl = (q_1, \dots, q_n)$ be a configuration in a given 
round and let $X$ be the random variable indicating the number of 
empty bins in the next round. For any large enough $n$, it holds that
$$
\Prob{}{X \leqslant \frac{n}{4}} \leqslant e^{- \alpha n},
$$
where $\alpha$ is a suitable positive constant.
\end{lemma}
\begin{proof}
Let $a = a(\confvl)$ and $b = 
b(\confvl)$ respectively denote the number of empty bins and the number 
of bins with exactly one token in 
configuration $\confvl$. For each bin $u$ of the $a+b$ bins with at 
most one token, let $Y_u$ be the random variable indicating whether 
or not bin $u$ is empty in the next round, so that
\[
X = \sum_{u = 1}^{a+b} Y_u 
\quad \mbox{ and } \quad
\Prob{}{Y_u = 1} = \left( 1 - \frac{1}{n} \right)^{n - a} \geqslant e^{-\frac{n-a}{n-1}},
\]
where in the last inequality we used the fact that $1 - x \geqslant e^{- \frac{x}{1-x}}$. 
Hence we have that
\begin{equation}\label{eq:expecemptyfirst}
\Expec{}{X} \geqslant (a+b) \, e^{-\frac{n-a}{n-1}}
\end{equation}
The crucial fact is that the number of bins with two or more tokens 
cannot exceed the number of empty bins, i.e. $n - (a+b) \leqslant 
a$. Thus, we can bound the number of empty bins from below\footnote{Observe 
that this argument only works to get a \emph{lower} bound on the number of 
empty bins and not for an upper bound.}, $a 
\geqslant (n-b)/2$, and by using that bound in~\eqref
{eq:expecemptyfirst} we get
\[
\Expec{}{X} \geqslant \frac{n+b}{2} \, e^{ -\frac{n+b}{2(n-1)} }
\]
Now observe that, for large enough $n$ a positive constant $\varepsilon$ exists such that
\[
\frac{n+b}{2} \, e^{ -\frac{n+b}{2(n-1)} } \geqslant (1+\varepsilon)\frac{n}{4}
\]
for every $0 \leqslant b \leqslant n$.

It is not difficult to prove that random variables $Y_1, \dots, 
Y_{a+b}$ are \emph{negatively associated} (e.g., see Theorem~13 in~
\cite{DR98}). Thus we can apply (see Lemma~7 in~\cite{DR98}) 
the Chernoff bound~(\ref{CB:lowertail}) with $\delta = \varepsilon 
/ (1+\varepsilon)$ to r.v. $X$ to obtain
\[
\Prob{}{X \leqslant \frac{n}{4}} \leqslant \exp\left(- \frac{\varepsilon^2}{4 (1+\varepsilon)} n \right)
\]
\end{proof}

\smallskip\noindent
From the above lemma it easily follows that, if we look at our 
process over a time-window $T = T(n)$ of polynomial size, after the 
first round we always see at least $n/4$ empty bins, w.h.p. More 
formally, for every $t \in \{1, \dots, T\}$, let $\mathcal{E}_t$ be the event ``The number of 
empty bins at round $t$ is at least $n/4$''. From Lemma \ref{lemma:emptying} 
and the union bound we get the following lemma.

\begin{lemma}\label{prop:emptyqueues}
Let $\confvl_0$ denote the initial configuration, let $T = T(n) = n^c$ 
for an arbitrarily large constant $c$. For any large enough $n$
it holds that
$$
\Prob{}{\bigcap_{t = 1}^T \mathcal{E}_t \;|\; \confvr^{(0)} = \confvl_0} \geqslant 1 - e^{-\gamma n}
$$
where $\gamma$ is a suitable positive constant.
\end{lemma}

\begin{proof}
By using the union bound we have that
\[
\Prob{}{\bigcap_{t = 1}^T \mathcal{E}_t \;|\; \confvr^{(0)} = \confvl_0}
= 1 - \Prob{}{\bigcup_{t = 1}^T \overline{\mathcal{E}_t} \;|\; \confvr^{(0)} = \confvl_0} 
\geqslant  1 - \sum_{t = 1}^T \Prob{}{\overline{\mathcal{E}_t} \;|\; \confvr^{(0)} = \confvl_0}
\]
By conditioning on the configuration at round $t-1$, from the Markov property and Lemma~\ref{lemma:emptyqueuesround} it then follows that
\[
\Prob{}{\overline{\mathcal{E}_t} \;|\; \confvr^{(0)} = \confvl_0} 
= \sum_{\confvl} \Prob{}{\overline{\mathcal{E}_t} \;|\; \confvr^{(t-1)} 
= \confvl} \Prob{}{\confvr^{(t-1)} = \confvl \;|\; \confvr^{(0)} = \confvl_0} 
\leqslant e^{-\alpha n}
\]
Hence,
\[
\Prob{}{\bigcap_{t = 1}^T \mathcal{E}_t \;|\; \confvr^{(0)} = \confvl_0} 
\geqslant 1 - T e^{-\alpha n} 
\geqslant 1 - e^{-\gamma n}
\]
for a suitable positive constant $\gamma$.
\end{proof}

\subsection{Coupling with \tetris\ }
Using a coupling argument and Lemma~\ref{prop:emptyqueues} we now prove that 
the maximum load in the original process is stochastically not larger than the
maximum load in the \tetris\ process w.h.p.

In what follows we denote by $W^{(t)}$ the \emph{set} of non-empty bins at 
round $t$ in the original process. Recall that, in the latter, at every round  
a ball is selected from every non-empty bin $u$  and it is moved to a bin 
chosen u.a.r. Accordingly we define, for every round $t$, the random variables
\begin{equation}\label{eq:movingtokens}
\left\{ X_u^{(t+1)} \,:\, u \in W^{(t)} \right\},
\end{equation}
where $X_u^{(t+1)}$ indicates the new position reached in round $t+1$
by the ball selected in round $t$ from bin $u$. Notice that for 
every non-empty bin $u \in W^{(t)}$ we have that 
$\Prob{}{X_u^{(t+1)} = v} = 1/n$ for every bin 
$v \in [n]$. The random process $\left\{ \confvr^{(t)} \,:\, t \in 
\mathbb{N} \right\}$ is completely defined by random variables 
$X_u^t$'s, indeed we can write
\[
\queue_v^{(t+1)}  =  \queue_v^{(t)} \dotdiv 1 + \left|\left\{ u \in 
W^{(t)} \,:\, X_u^{(t+1)} = v \right\} \right|\quad \mbox{\emph{and}}\quad 
W^{(t+1)}  =  \left\{ u \in [n] \,:\, \queue_u^{(t+1)} \geqslant 1 \right\},
\]
where we used notation $a \dotdiv b = \max\{a-b,0\}$.
Analogously, for each bin $u \in [n]$ in the \tetris\ process, let $\hat{\queue}_u^{(t)}$ be the 
random variable indicating the number of balls in bin $u$ in round 
$t$. We next prove that, over any polynomially-large time 
window, the maximum load of any bin in our process is stochastically 
smaller than the maximum number of balls in a bin of the \tetris\ 
process w.h.p. More formally, we prove the following lemma. 

\begin{lemma}\label{le:tetris}
Assume we start our process and the \tetris\ process from the same 
initial configuration $\confvl = (q_1, \dots, q_n)$ such that 
$\sum_{u = 1}^n q_u = n$ and containing at least $n/4$ empty bins. Let 
$T = T(n)$ be an arbitrary round and let $M_T$ and $\hat{M}_T$ be 
respectively the random variables indicating the maximum loads in our 
original process and in the \tetris\ process, up to round $T$. 
Formally
\begin{align*}
M_T = \max\{ \queue_u^{(t)} \;:\, u \in [n], \, t = 1, 2, \dots, T \} \\[2mm]
\hat{M}_T = \max\{ \hat{\queue}_u^{(t)} \;:\, u \in [n], \, t = 1, 2, \dots, T \}
\end{align*}
For every $k\geq 0$ it holds that
$$
\Prob{}{M_T \geqslant k} \leqslant \Prob{}{\hat{M}_T \geqslant k} + T \cdot e^{-\gamma n}
$$
for a suitable positive constant $\gamma$.
\end{lemma}

\begin{proof}
We proceed by coupling the \tetris\ process with the original one round by 
round. Intuitively speaking the coupling proceeds as follows:\\
- Case (i): \emph{the number of non-empty bins in the original process
is $k\leq \frac 34 n$}. For each non-empty bin $u$, let $i_u$ be the ball 
picked from $u$. We throw one of the  $\frac 34 n$ new balls of the \tetris\ 
process in the same bin in which $i_u$ ends up.
Then, we throw all the remaining $\frac 34 n -k$ balls independently u.a.r.\\
- Case (ii): \emph{the number of non-empty bins is $k >\frac 34 n$}.
We run one round of the \tetris\ process independently from the original one.

\noindent
By construction, if the number of non-empty bins in the original process
is not larger than $\frac 34 n$ at any round, case (ii) never applies 
and the \tetris\ process ``dominates" the original one, meaning that every bin in the \tetris\ process
contains at least as many balls as the corresponding bin in the original one.
Since from Lemma \ref{prop:emptyqueues} we know that the number of non-empty 
bins in the original process is not larger than $\frac 34 n$ for any 
time-window of polynomial size w.h.p., we thus have that the \tetris\ process 
dominates the original process for the whole time window w.h.p.

\smallskip
More formally, for $t \in\{1, \dots, T\}$, denote by $B^{(t)}$ the set of new 
balls in the \tetris\ process at round $t$ (recall that the size of $B^{(t)}$ 
is $(3/4)n$ for every $t \in\{1, \dots, T\}$). For any round $t$ and any 
ball $i \in B^{(t)}$, let $\hat{X}_i^{(t)}$ be the random variable 
indicating the bin where the ball ends up. Finally, let $\left\{ 
U_i^{(t)} \;:\; t = 1, \dots, T, \, i \in B^{(t)} \right\}$ be a 
family of i.i.d. random variables uniform over $[n]$.

\smallskip
At any round $t \in \{ 1, \dots, T\}$:

\smallskip\noindent
\underline{If $|W^{(t-1)}| \leqslant (3/4)n$}: Let $B^{(t)}_W$ be an 
arbitrary subset of $B^{(t)}$ with size exactly $|W^{(t-1)}|$, let 
$f^{(t)} \;:\; B^{(t)}_W \rightarrow W^{(t-1)}$ be an arbitrary 
bijection and set
\begin{equation}
\hat{X}_i^{(t)} = 
\left\{
\begin{array}{cl}
X_i^{(t)} & \quad \mbox{ if } i \in B^{(t)}_W \\[2mm]
U_i^{(t)} & \quad \mbox{ if } i \in B^{(t)} \setminus B^{(t)}_W
\end{array}
\right.
\label{eq:indicator_in_tetris}
\end{equation}

\smallskip\noindent
\underline{If $|W^{(t-1)}| > (3/4)n$}: Set $\hat{X}_i^{(t)} = U_i^{(t)}$ for all $i \in B^{(t)}$.

\smallskip\noindent
By construction we have that random variables
$$
\left\{ \hat{X}_i^{(t)} \;:\; t \in\{ 1,2, \dots, T\}, \, i \in B^{(t)} \right\}
$$
are mutually independent and uniformly distributed over $[n]$. Moreover, in the 
joint probability space for any $k$ we have that
\[
\Prob{}{M_T \geqslant k} 
= \Prob{}{M_T \geqslant k, \, \hat{M}_T \geqslant M_t } 
+ \Prob{}{M_T \geqslant k, \, \hat{M}_T < M_T} 
\leqslant \Prob{}{\hat{M}_T \geqslant k} + \Prob{}{\hat{M}_T < M_T}
\]
Finally, let $\mathcal{E}_T$ be the event ``There are at least $n/4$ 
empty bins at all rounds $t \in\{ 1, \dots, T\}$'' and observe that, from 
the coupling we have defined, the event 
$\mathcal{E}_T$ implies event ``$\hat{M}_T \geqslant M_T$''. Hence 
$\Prob{}{\hat{M}_T < M_T} \leqslant \Prob{}{\overline{\mathcal{E}_T}}$ 
and the thesis follows from Lemma~\ref{prop:emptyqueues}.
\end{proof}

\subsection{Analysis of the \tetris\ process}
We begin by observing that in the \tetris\ process, the random variables indicating the number 
of balls ending up in a bin in different rounds are i.i.d. binomial. 
This fact is extremely useful to give upper bounds on the load of 
the bins, as we do in the next simple lemma, that will be used to 
prove self-stabilization of the original process.

\begin{lemma}\label{lemma:emptying}
From any initial configuration, in the \tetris\ process every bin will 
be empty at least once within $5 n$ rounds, w.h.p.
\end{lemma}
\proof
Let $u \in [n]$ be a bin with $k \leqslant n$ balls in the initial 
configuration. For $t \in\{ 1, \dots, 5n\}$ let $Y_t$ be the random 
variable indicating the number of new balls ending up in bin $u$ at 
round $t$. Notice that in the \tetris\ process $Y_1, \dots, Y_{5n}$ 
are i.i.d. $B\left((3/4)n, \, 1/n \right)$ hence $\Expec{}{Y_1 + 
\cdots + Y_{5n}} = (15/4) n$ and by applying Chernoff bound~(\ref
{CB:uppertail}) with $\delta = 1/15$ we get
$$
\Prob{}{Y_1 + \cdots + Y_{5n} \geqslant 4n} \leqslant e^{- \alpha n}
$$
where $\alpha = 1/(180)$.

\smallskip\noindent
Now let $\mathcal{E}_u$ be the event \emph{``Bin $u$ will be 
non-empty for all the $5n$ rounds''}. Since when a bin is non-empty 
it looses a ball at every round, event $\mathcal{E}_u$ implies, in 
particular, that
$$
k - 5n + Y_1 + \cdots + Y_{5n} \geqslant 0
$$
That is $Y_1 + \cdots + Y_{5n} \geqslant 5n - k \geqslant 4n$. Thus
$$
\Prob{}{\mathcal{E}_u} \leqslant \Prob{}{Y_1 + \cdots + Y_{5n} \geqslant 4n} \leqslant e^{-\alpha n}
$$
The thesis follows from the union bound over all bins $u \in [n]$.
\qed

We next focus on the maximum load that can be observed in the \tetris\ process 
at any given bin within a finite interval of time. We 
note that this result could be proved using tools from {\em drift 
analysis} (e.g., see \cite{hajek1982hitting}). We provide here an 
elementary and direct proof, that explicitely relies on the Markovian 
structure of the \tetris\ process.

Let $\{X_t\}_t$ be a sequence of i.i.d. $B\left((3/4)n,1/n \right)$ random 
variables and let $Z_t$ be the Markov chain with state space 
$\{0,1,2, \dots \}$ defined as follows
\begin{equation}\label{eq:absorbingMC}
Z_{t} = 
\left\{
\begin{array}{cl}
0 & \mbox{ if } Z_{t-1} = 0 \\[2mm]
Z_{t-1} - 1 + X_t & \mbox{ if } Z_{t-1} \geqslant 1
\end{array}
\right.
\end{equation}
Observe that $0$ is an absorbing state for $Z_t$ and let $\tau$ be the 
absorption time $\tau = \inf \{t \in \mathbb{N} \,:\, Z_t = 0\}$. We first 
prove the following lemma.
\begin{lemma}\label{lemma:absorbingMC}
For any initial starting state $k \in \mathbb{N}$ and any $t \geqslant 8 k$, it
holds that
\[
\Prob{k}{\tau > t} \leqslant e^{- t / 144}
\]
\end{lemma}
\begin{proof}
Observe that
\[
\Prob{k}{\tau > t} 
= \Prob{k}{Z_t > 0} 
= \Prob{}{k + \sum_{i = 1}^t X_i - t > 0}
= \Prob{}{\sum_{i = 1}^{t} X_i > t - k}
\leqslant \Prob{}{\sum_{i = 1}^{t} X_i > \frac{7}{8} t}
\]
where in the last inequality we used hypothesis $k < (1/8)t$.
Since the $X_i$s are i.i.d. binomial $B((3/4)n, 1/n)$, it follows that
$\sum_{i = 1}^{t} X_i$ is binomial $B((3/4)nt, 1/n)$ and from Chernoff bound
we have that
\[
\Prob{}{\sum_{i = 1}^t X_i > \frac{7}{8} t} 
= \Prob{}{\sum_{i = 1}^t X_i > \left( 1 + \frac{1}{6}\right) \frac{3}{4}t}
\leqslant e^{- \frac{(1/6)^2}{3} \frac{3}{4}t} 
= e^{- t/144}
\]
\end{proof}

\smallskip\noindent
Now we can easily prove the following statement on the \tetris\ process.

\begin{lemma}\label{thm:hp}
Let $c$ be an arbitrarily-large constant, and let the \tetris\ process start 
from any legitimate configuration. The maximum load $\hat M^{(t)}$ is  
$\bigO(\log n)$ for all $t = \bigO(n^c)$, w.h.p. 
\end{lemma}
\begin{proof}
Consider an arbitrary bin $u$ that is non-empty in the initial legitimate 
configuration. Let $\hat\queue^{(0)} = \bigO(\log n)$ be its initial 
load\footnote{We omit the subscript $u$ in the remainder of this proof since clear 
from context.} and let 
$\tau = \inf\left\{ t \,:\, \hat\queue^{(t)} = 0 \right\}$ be the first round the
bin becomes empty. Observe that, for any $t \leqslant \tau$, $\hat\queue^{(t)}$ 
behaves exactly as the Markov chain defined in~\eqref{eq:absorbingMC}. Hence,
from Lemma~\ref{lemma:absorbingMC} it follows that for every constant $\hat{c}$ 
such that $\hat{c} \log n \geqslant 8 \hat\queue^{(0)}$ we have
\begin{equation}\label{eq:binphasemaxload}
\Prob{\hat\queue^{(0)}}{\tau > \hat{c} \log n} \leqslant n^{-\hat{c}/144}
\end{equation}
Thus, within $\bigO(\log n)$ rounds the bin will be empty w.h.p., and since the
load of the bin decreases of at most one unit per round, the load of the
bin is $\bigO(\log n)$ for all such rounds w.h.p.

Next, define a \textit{phase} as any sequence of rounds that starts when the bin
becomes non-empty and ends when it becomes empty again. Notice that, by using a
standard balls-into-bins argument, in the first round of each phase the load of
the bin will be $\bigO(\log n / \log \log n)$ w.h.p. Moreover, in any phase the
load of the bin can be coupled with the Markov chain in~\eqref{eq:absorbingMC}.
Hence, for any arbitrary large constant $c$ we can choose the constant $\hat{c}$ 
in~\eqref{eq:binphasemaxload} large enough so that, by taking the union bound
over all phases up to round $n^c$, the load of the bin is $\bigO(\log n)$ in 
all rounds $t \leqslant n^c$ w.h.p.

Finally, observe that for any bin that is initially empty the same argument 
applies with the only difference that the first phase for the bin does not start
at round $0$ but at the first round the bin becomes non-empty. The thesis thus
follows from a union bound over all the bins.
\end{proof}

\subsection{Back to the original process: Proof of Theorem~\ref{thm::main}}
From a standard balls-into-bins argument (see, e.g., \cite{MU05}), 
starting from any legitimate configuration, after one round the process 
still lies in a legitimate configuration w.h.p. and, thanks to 
Lemma~\ref{lemma:emptyqueuesround}, there are at least $n/4$ empty bins 
w.h.p. From Lemma~\ref{le:tetris} with $T = \bigO\left( n^c \right)$,
we have that the maximum load of the repeated balls-into-bins process 
does not exceed the maximum load of the \tetris\ process in all rounds 
$1,\dots,T$, w.h.p. Finally, the upper bound on the maximum load of the \tetris\ 
process in Lemma~\ref{thm:hp} completes the proof of the first 
statement of Theorem~\ref{thm::main}.

As for self-stabilization, given an arbitrary initial configuration, Lemma~\ref
{lemma:emptying} implies that within $\bigO(n)$ rounds, all bins 
have been emptied at least once, w.h.p. When a bin becomes empty, Lemma
\ref{lemma:absorbingMC} ensures that its load will be $\bigO(\log 
n)$ over a polynomial number of rounds. Hence, within $\bigO(n)$
rounds, the system will reach a legitimate configuration, w.h.p.
\qed

\section{Parallel Resource Assignment}
As mentioned in the introduction, the repeated balls-into-bins 
process    can  also be seen as running parallel random walks of $n$
distinct tokens (i.e. balls), each of them starting from a    node 
(i.e. bins) of the complete graph of size $n$. This is a randomized 
protocol for the parallel allocation problem where tokens represent 
different resources/tasks  that must be assigned to all nodes in 
mutual exclusion \cite{C11}. In this scenario, a critical 
complexity measure is the (global) cover time, i.e., the time 
required by any token to visit all nodes.   
 
\noindent  It is important to observe that our analysis of the maximum 
load  works for anonymous  tokens and nodes and, 
hence, for any particular queuing  strategy. Under FIFO strategy,
no token spends in a bin a number of rounds exceeding the current load 
as it entered the bin. Theorem~\ref{thm::main} then implies  that, 
after an initial stabilizing phase of $\bigO(n)$ rounds, every token 
will spend at most a logarithmic number of rounds in any bin queue 
it traverses and over any period of polynomial length, w.h.p. 
We also know that  the  cover time   of the  single 
random-walk process  is  w.h.p. $\bigO(n \log n)$ (see, e.g., 
\cite{MU05}).    
Combining the above two facts, we easily get the following, almost 
tight result on the Parallel Resource Assignment problem.
  
\begin{corollary} \label{cor:covtime}
The random-walk protocol for the Parallel Resource Assignment problem on the clique has  cover time
$\bigO\left(n \log^2n\right)$, w.h.p. 
\end{corollary}
  
\subsubsection*{Adversarial model.}	 
The self-stabilization property shown in Theorem \ref{thm::main} makes the random walk protocol robust
to transient faults. We can   consider an adversarial model in which, in some \emph{faulty} rounds, an adversary
can reassign the tokens to the nodes in an arbitrary way. 
Then, the linear convergence time  shown in Theorem~\ref{thm::main}  implies that the  $\bigO\left(n \log^2n\right)$ bound on the cover time 
still   holds provided the faulty rounds happen with a frequency not higher than $\gamma n$, 
for any constant $\gamma\geq 6$. 
Indeed, thanks to Lemma \ref{lemma:emptying}, the action of an adversary manipulating the system configuration once every $\gamma n$ 
rounds can affect only the successive $5n$ rounds, while  our analysis in the non-adversarial model does hold for the remaining $(\gamma-5) n$ rounds. 
It follows that the overall slowdown on the cover time produced by such an adversary  
is at most a constant factor on the previous $\bigO\left(n\log^2 n\right)$ 
upper bound, w.h.p.

\section{Conclusions and Open Questions}

In this paper, we showed that repeated balls-into-bin  is self-stabilizing  
when the number $m$ of balls equals the number $n$ of bins (obviously, this 
is still the case, whenever $m < n$). An interesting  open  
question is whether this result extends to larger values of $m$, i.e., 
for any $m = \bigO(n \log n)$. We believe an approach based on a lower bound on the number of 
empty bins might still work. Simulation results for increasing values 
of $n$ (up to $n \sim 10^5$) show that the number of empty bins is 
still compatible with a linear function, even if standard deviation in our experiments 
turned out to be relatively large.
  
A more general interesting question is the study of this process 
over more general graph  classes. This line of research is also 
motivated by  several recent applications of parallel random walks 
in the (uniform) gossip model \cite{BCE10,C11,EK15,HPP12}. As 
mentioned in the introduction, previous analysis of this process 
provides a bound $\bigO\!\left(\sqrt t\right)$ on the maximum load after $t$ rounds on regular 
graphs \cite{BCN14}. We believe this previous bound for regular graphs 
is far from tight and it leads to rough bounds on parallel 
cover times on these networks. We conjecture  that   the   maximum load  remains 
logarithmic for a long period in any \emph{regular} graph. A 
possible reason for this phenomenon (if true) might be that the  expected 
difference  between (token) arrivals and departures is always 
\emph{non-positive} at every node in regular graphs. As highlighted in 
our analysis of the complete graph, this fact alone is 
not enough but it could be combined with a suitable bound on the 
number of empty bins, in order to prove our conjecture in this more 
general case. Unfortunately, non-complete  graphs present a further 
technical issue:  in order to apply any argument based on the 
presence of empty bins, not only do we need to argue about their 
number, but also about their distribution across the network. This     
technical issue seems to be far from trivial even on simple  topologies such as rings.

Finally, a technical question concerns the tightness of our bound on the
maximum load. In the classical 
(one shot) balls-into-bins problem, it is well-known that the 
maximum load of the bins is $\Theta\left(\log n / \log\log n \right)$ w.h.p.
One may wonder whether our $\bigO\left(\log n \right)$ upper bound on the maximum load
of the repeated process for a polynomial number of rounds is tight, or it can 
be improved to $\bigO\!\left(\log n / \log \log n \right)$. We 
conjecture that, within any polynomial time window, the probability 
that the maximum load asymptotically exceeds $\log n / \log \log n$ is 
non-negligible.

\subsection*{Acknowledgments}
We would like to thank Riccardo Silvestri for helpful discussions and important hints.

\bibliographystyle{abbrv}
\bibliography{congestion}  

\bigskip
\begin{center}
\begin{LARGE}
\textbf{Appendix}
\end{LARGE}
\end{center}
\appendix
\section{Useful inequalities}
\begin{lemma}[Chernoff bound]
Let $\{X_t \,:\, t \in [n] \}$ be a family of independent binary 
random variables. Let $X = \sum_{t = 1}^n X_t$ and let $\mu_L 
\leqslant \Expec{}{X} \leq \mu_H$. For every $\delta \in (0,1)$ it holds that
\begin{align}
\Prob{}{X \leqslant (1-\delta)\mu_L} & \leqslant  \exp\left(-\frac{\delta^2}{2} \mu_L \right) \label{CB:lowertail} \\[2mm]
\Prob{}{X \geqslant (1+\delta)\mu_H} & \leqslant  \exp\left(-\frac{\delta^2}{3} \mu_H \right) \label{CB:uppertail} 
\end{align}
\end{lemma}

\section{Negative association}\label{sec::apx-association}
\begin{definition}[Negative association]
Random variables $X_1, \dots, X_n$ are \emph{negatively associated} if, for every pair of disjoint subsets $I, J \subseteq [n]$, it holds that
\[
\Expec{}{f \left(X_i, \, i \in I \right) \cdot g \left(X_j, \, j \in J \right)} 
\leqslant 
\Expec{}{f \left(X_i, \, i \in I \right)} \cdot \Expec{}{g \left(X_j, \, j \in J \right)}
\]
for all pairs of functions $f\,:\, \mathbb{R}^{|I|} \rightarrow \mathbb{R}$ and $g\,:\, \mathbb{R}^{|J|} \rightarrow \mathbb{R}$ that are both non-decreasing or both non-increasing.
\end{definition}

Now we give a simple counterexample showing that, in our balls-into-bins process, the random variables counting the number of balls arriving in a given bin in different rounds cannot be negatively associated.

Consider our random process with $n=2$ and let $X_1$ and $X_2$ be the random variables indicating the number of tokens arriving at the first bin in rounds $1$ and $2$, respectively. Let $f \equiv g$ be the non-increasing function
$$
f(x) =
\left\{
\begin{array}{cl}
1 & \quad \mbox{ if } x = 0\\[2mm]
0 & \quad \mbox{ if } x > 0 
\end{array}
\right.
$$
If $X_1$ and $X_2$ were negatively associated, we thus would have that $\Prob{}{X_1 = 0,\, X_2 = 0} \leqslant \Prob{}{X_1 = 0} \Prob{}{X_2 = 0}$. However, by direct calculation it is easy to compute that
$$
\Prob{}{X_1 = 0, \, X_2 = 0} = 1/8 
$$
because, in order for ``$X_1 = 0, \, X_2 = 0$" to happen, at the first round both balls have to end up in the second bin (this happens with probability $1/4$) and at the second round the ball chosen in the second bin has to stay there (this happens with probability $1/2$). But we have that $\Prob{}{X_1 = 0} = 1/4$ and by conditioning on all the three possible configurations at round $1$ we have $\Prob{}{X_2 = 0} = 3/8$. Thus
$$
\frac{1}{8} = \Prob{}{X_1 = 0,\, X_2 = 0} > \Prob{}{X_1 = 0} \Prob{}{X_2 = 0} = \frac{1}{4} \cdot \frac{3}{8}
$$

\smallskip\noindent
In general, intuitively speaking it seems that event ``$X_t = 0$'' makes more likely the event that there are a lot of empty bins in the system, which in turn makes more likely event ``$X_{t+1} = 0$'' that the bin will receive no tokens at round $t+1$ as well.

\end{document}